\documentclass[conference,10pt]{IEEEtran}

\makeatletter
\def\ps@headings{%
\def\@oddhead{\mbox{}\scriptsize\rightmark \hfil \thepage}%
\def\@evenhead{\scriptsize\thepage \hfil\leftmark\mbox{}}%
\def\@oddfoot{}%
\def\@evenfoot{}}
\makeatother
\usepackage{cite}
\usepackage{amssymb,amsmath}
\usepackage{algorithm}
\usepackage{algorithmic}
\usepackage{graphicx}
\usepackage[tight,footnotesize]{subfigure}
\usepackage{color}
\usepackage{url}

\newtheorem{definition}{Definition}
\newtheorem{proposition}{Proposition}
\newtheorem{proof}{Proof}

\begin{document}

\title{On Scheduling and Redundancy for P2P Backup}

\author{\IEEEauthorblockN{Laszlo Toka\IEEEauthorrefmark{1}\IEEEauthorrefmark{2}, Matteo Dell'Amico\IEEEauthorrefmark{1}, Pietro Michiardi\IEEEauthorrefmark{1}\vspace{3pt}\\\{laszlo.toka, matteo.dell-amico, pietro.michiardi\}@eurecom.fr\\}
\IEEEauthorblockA{\IEEEauthorrefmark{1} Eurecom, Sophia-Antipolis, France}
\IEEEauthorblockA{\IEEEauthorrefmark{2} Budapest University of Technology and Economics, Hungary}}

\maketitle

\begin{abstract}

An online backup system should be quick and reliable in both saving
and restoring users' data. To do so in a peer-to-peer implementation,
data transfer scheduling and the amount of redundancy must be chosen
wisely. We formalize the problem of exchanging multiple pieces of data
with intermittently available peers, and we show that random
scheduling completes transfers nearly optimally in terms of duration
as long as the system is sufficiently large. Moreover, we propose an
adaptive redundancy scheme that improves performance and decreases
resource usage while keeping the risks of data loss low.  Extensive
simulations show that our techniques are effective in a realistic
trace-driven scenario with heterogeneous bandwidth.

\end{abstract}

\section{Introduction}

The advent of cloud computing as a new paradigm to enable service providers with the ability to deploy cost-effective solutions has favored the development of a range of new services, including online storage applications. Due to the economy of scale of cloud-based storage services, the costs incurred by end-users to hand over their data to a remote storage location in the Internet have approached the cost of ownership of commodity storage devices.

As such, online storage applications spare users most of the
time-consuming nuisance of data backup: user interaction is minimal,
and in case of data loss due to an accident, restoring the original
data is a seamless operation.  However, the long-term storage costs that
are typical of a backup application may easily go past that of
traditional approaches to data backup. Additionally, while data
availability is a key feature that large-scale data-centers
deployments guarantee, its durability is questionable, as reported
recently\cite{carbonite}.

For these reasons, peer-to-peer (P2P) storage systems are an alternative to cloud-based solutions. Storage costs are merely those of a commodity storage device, which is shared (together with some bandwidth resources) with a number of remote Internet users to form a distributed storage system. Such applications optimize latency to individual file access: indeed, users hand over their data to the P2P system, which is used as a replacement of a local hard drive. In such a scenario, low access latency is difficult to achieve: the online behavior of users is unpredictable and, at large scale, crashes and failures are the norm rather than the exception. As a consequence, storage space is sacrificed for low access latency: a P2P application stores large amounts of redundant data to cope with such unfavorable events.

In this work we study a particular case of online storage: P2P backup
applications. Data backup involves the bulk transfer of potentially
large quantities of data, both during regular data backups and in case
of data loss. As a consequence, low access latency is not an issue,
while short backup and restore times seem a more reasonable goal.

Given these considerations, here we seek to optimize backup
and restore times, while guaranteeing that data loss remains an
unlikely event. There are two main design choices that affect these
metrics: \emph{scheduling}, i.e. deciding how to allocate data
transfers between peers, and \emph{redundancy}, i.e. the amount of
data in the P2P system that guarantees a backup operation to be
considered complete and safe. The endeavor of this work is to study
and evaluate these two intertwined aspects.

First, we describe in detail our application scenario (Sec.   \ref{sec:background}), and show why the assumptions underlying a backup application can simplify many problems addressed in the literature.
We then set off to define the problem of scheduling in a full knowledge setting, and we show that it can be solved in polynomial time by reducing it to a maximal flow problem. Full knowledge of future peer uptime is obviously an unrealistic assumption: thus, we show that a randomized approach to scheduling yields near optimal results when the system scale is large and we corroborate our findings using real availability traces from an instant messaging application (Sec. \ref{sec:scheduling}).

We then move to study a novel redundancy policy that, rather than focusing on short-term data availability, targets short data restore times. As such, our method alleviates the storage burden of large amounts of redundant data on client machines (Sec. \ref{sec:red}). With a trace-driven simulation of a complete P2P backup system, we show that our technique is viable in practical scenarios and illustrate its benefits in terms of increased performance (Sec. \ref{sec:complex_sim}).

We conclude by studying a range of data maintenance policies when
restore operations may undergo some natural delays. For example,
detecting a faulty external hard-drive may not be immediate, or
obtaining a new equipment upon a crash may require some time. We show
that an ``assisted'' approach to data repair techniques (which
involves a cloud-based storage service) can significantly reduce the
probability of data loss, at an affordable cost (Sec. \ref{sec:repair}).

\section{Application Scenario}
\label{sec:background}

In this work, similarly to many online backup applications (e.g.,
Dropbox\footnote{\url{https://www.dropbox.com/}}), we assume users to
specify one local folder containing important data to backup. We also
assume that backup data remains available locally to peers. This is an
important trait that distinguishes backup from storage applications,
in which data is only stored remotely.

Backup data consists of an opaque object, possibly representing an
encrypted archive of changes to a set of files, that we
term \emph{backup object}. In the spirit of incremental backups, we
consider that each backup object should be kept on the system
indefinitely. Consolidation and deletion of obsolete backups are not
taken into account in this work.

A backup object of size $o$ is split into $k$ original fragments of a
fixed size $f$, with $k=o / f$. Since backup data is stored on
unreliable machines characterized by an unpredictable online behavior,
the original $k$ blocks are encoded using erasure coding (e.g.,
Reed-Solomon). This creates $n$ \emph{encoded fragments} having size
$f$, of which any $k$ are sufficient to recover the original data. The
redundancy rate is defined as $r=n / k$. Here we assume that encoded
fragments reside on \emph{distinct} remote peers, which avoids that a
single disk failure causes the loss of multiple fragments.

\noindent \textbf{Backup Phase:} The backup phase involves a data
owner and a set of remote peers that eventually store encoded
fragments for the data owner. We assume that any peer in the system
can collect a list of remote peers with available storage space: this
can be achieved by using known techniques, e.g. a centralized
``tracker'' or a decentralized data structure such as a distributed
hash table.

Data backup requires a \emph{scheduling policy} that drives the choice
of where and when to upload encoded fragments to remote
peers. Moreover, a \emph{redundancy policy} determines when the data is safe, which completes the backup operation.

\noindent \textbf{Maintenance Phase:} Once the backup phase is
completed and encoded fragments reside on remote peers, the
maintenance phase begins. Peer crashes and departures can cause the
loss of some encoded fragments; during the maintenance phase, peers
detects such losses and generate new encoded fragments to restore a
redundancy level at which the backup is safe again.

For a generic P2P storage system, in which encoded fragments only
reside in the network and peers do not keep a local copy of their
data, the maintenance phase is critical. Indeed, peers need to first
download the whole backup object from remote machines, then to
generate new encoded fragments and upload them to available
peers. This problem has fostered the design of efficient coding
schemes to mitigate the excessive network traffic caused by the
maintenance operation (see
e.g.~\cite{dimakis07,duminuco-biersack-08}).

In a backup application, the maintenance phase is less critical: the
data owner can generate new encoded fragments using the local copy of
the data with no download required.

\noindent \textbf{Restore Phase: } In the unfortunate case of a
crash, the data owner initiates the restore phase. A peer contacts the
remote machines holding encoded fragments, downloads at least $k$ of
them, and reconstructs the original backup data. Again, a scheduling
policy drives the process.

Since the ability to successfully restore data upon a crash is the
ultimate goal of any backup system, in our application the restore
traffic receives higher priority than the backup and maintenance
traffic.

\subsection{Performance Metrics}

We characterize the system performance in terms of the amount of time
required to complete the backup and the restore phases,
labelled \emph{time to backup} (TTB) and \emph{time to restore}
(TTR). 

In the following Sections, we use baselines for backup and restore
operations which bound both TTB and TTR. Let us assume an ideal
storage system with unlimited capacity and uninterrupted online time
that backs up user data. In this case, TTB and TTR only depend on
backup object size and on bandwidth and availability of the data
owner. We label these ideal values \emph{minTTB} and \emph{minTTR},
and we define them formally in Sec.~\ref{sec:scheduling}.

Additionally, we consider the \emph{data loss probability}, which
accounts for the probability of a data owner to be unable to restore
backup data.

A P2P backup application may exact a high toll in terms of peer
resources, including storage and bandwidth. In this work we gloss over
metrics of the burden on individual peers and the network, considering
a scenario in which the resources of peers are lost if left unused.

\subsection{Availability Traces}
\label{sec:availability_trace}

The \emph{online behavior} of users, i.e., their patterns of
connection and disconnection over time, is difficult to capture
analytically. In this work we will perform our evaluations on a real
application trace that exhibits both heterogeneity and correlated user
behavior. Our traces capture user availability, in terms of
login/logoff events, from an instant messaging (IM) server in Italy
for a duration of 3 months. We argue that the behavior of regular IM
users constitutes a representative case study. Indeed, for both an IM
and an online backup application, users are generally signed in for as
long as their machine is connected to the Internet.

In this work we only consider users that are online for an average of
at least four hours per day, as done in the Wuala online storage
application\footnote{\url{http://www.wuala.com/}}. Once this filter is
applied, we obtain the trace of 376 users. User availabilities are
strongly correlated, in the sense that many users connect or
disconnect around the same time. As shown in Fig. \ref{fig:avail2},
there are strong differences between the number of users connected
during day and night and between workdays and weekends. Most users are
online for less than 40\% of the trace, while some of them are almost
always connected (Fig. \ref{fig:a}).

\begin{figure}
\centering
	\subfigure[Online peers during a week.]{\label{fig:avail2} \includegraphics[width=.23\textwidth]{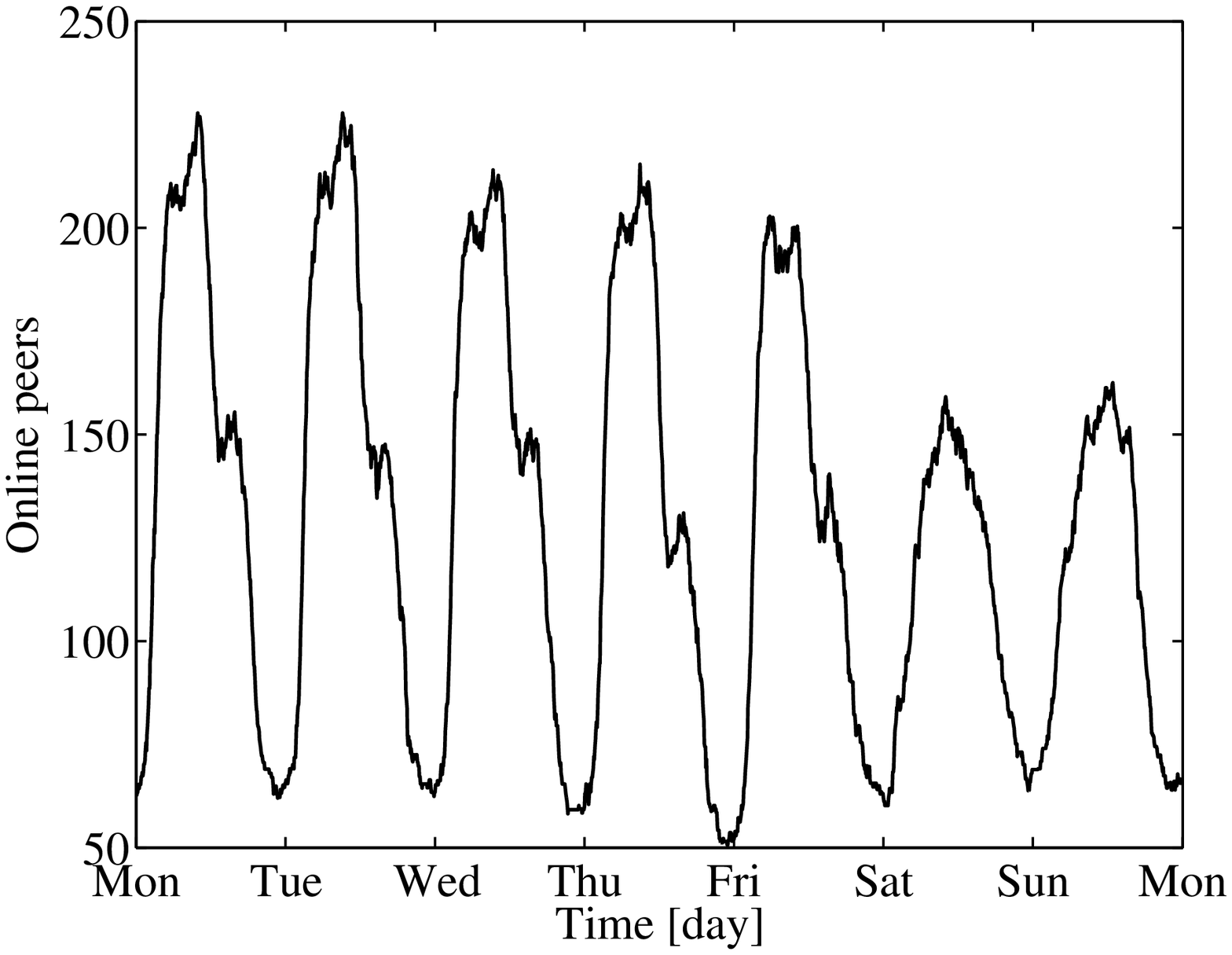}}
	\subfigure[Time spent online.]{\label{fig:a} \includegraphics[width=.23\textwidth]{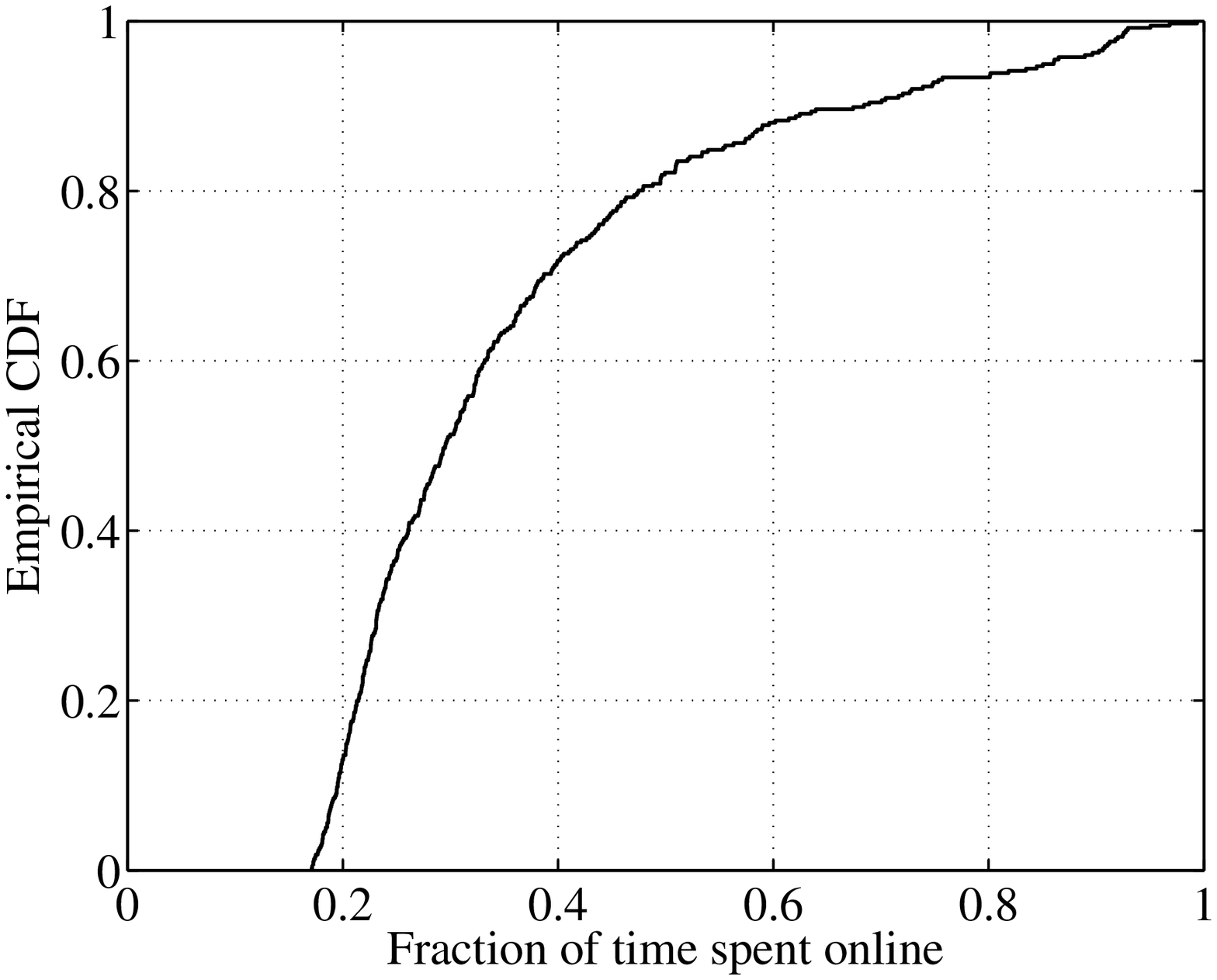}}	
\caption{Availability trace. Uptime is heterogeneous and strongly
  correlated.}
\label{fig:basics}
\end{figure}

\section{The Scheduling Problem}
\label{sec:scheduling}

Scheduling data transfers between peers is an important operation that
affects the time required to complete a backup or a restore task,
especially in a system involving unreliable machines with
unpredictable online patterns. Because of churn, a node might not be
able to find online nodes to exchange data with: hence, TTB and TTR
can grow due to idle periods of time. Unexpected node disconnections
require a method to handle partial fragments, which can be discarded
or resumed.  Moreover, the redundancy rate used to cope with failures
and unavailability may decrease system performance. Finally, the
available bandwidth between peers involved in a data transfer, which
may be shared due to parallel transmissions, is another cause for slow
backup and restore operations.

In this Section, we focus on the implications of churn alone. We
simplify the scheduling problem by assuming the redundancy factor to
be a given input parameter, and neglecting the possibility of
congestion due to several different backup, restore or maintenance
processes interfering. Furthermore, we do not consider interrupted
fragment transfers. In Section \ref{sec:red}, we define an adaptive
scheme to compute the redundancy rate applied to a backup operation
and in Section \ref{sec:complex_sim} we relax all other assumptions.

We now define a reference scenario to bound TTB and TTR. Consider an
ideal storage system (e.g. a cloud service) with unbounded bandwidth
and 100\% availability. A peer $i$ with upload and download bandwidth
$u_i$ and $d_i$ starting the backup of an object of size $o$ at time
$t$ completes its backup at time $t'$, after having spent $\frac o
{u_i}$ time online. Analogously, $i$ restores a backup object with the
same size at $t''$ after having spent $\frac o {d_i}$ time online. We
define $minTTB(i, t)=t' - t$ and $minTTR(i, t)=t'' - t$. We use these
reference values throughout the paper to compare the relative
performance of our P2P application versus that of such an ideal
system.

Because we neglect congestion issues, we can focus on a backup/restore
operation as seen from a single peer in the system. Let us consider a
generic peer $p_0$ and $I$ remote peers $p_1, \ldots, p_I$ used to
store $p_0$'s data. We assume time to be fractioned
in \emph{time-slots} of fixed length. Let $a_{i,t}$ be an indicator
variable so that $a_{i,t}=1$ if and only if $p_i$ is online at time
$t$. Each peer $i$ has integer upload and download capacity of
respectively $u_i$ and $d_i$ fragments per time-slot.

We now proceed with a series of definitions used to formalize the
scheduling problem.

\begin{definition}
	A \emph{backup schedule} is a set of $(i, t)$ tuples representing the decision of uploading a fragment from $p_0$ to peer $p_i$, where $i \in \{1 \ldots I\}$ at time-slot $t$. A valid backup schedule $S$ satisfies the following properties: 
	\begin{enumerate}
		\item $\forall t: \left|\left\{i : (i, t) \in S\right\}\right| \leq u_0$: no more than $u_0$ fragments per time-slot can be uploaded. 
		\item $\forall (i,t)\in S: a_{i,t}=a_{0,t}=1$: fragments are transferred only between online peers. 
		\item $\forall (i,t), (j, u) \in S: i \neq j$: no two fragments are stored on the same peer. 
	\end{enumerate}
	
	\label{def:backup_schedule} 
\end{definition}

\begin{definition}
	A \emph{restore schedule} is a set of $(i,t)$ tuples representing the decision of downloading a fragment from a set of remote peers $p_i \in P$ at time $t$, where $P$ is set of storage peers that received a fragment during the backup phase. A valid restore schedule $S$ satisfies the following properties: 
	\begin{enumerate}
		\item $\forall t: \left|\left\{i : (i, t) \in S\right\}\right| \leq d_0$: no more than $d_0$ fragments per time-slot can be recovered. 
		\item $\forall (i,t)\in S: a_{i,t}=a_{0,t}=1$. 
		\item $\forall (i,t): (j, u) \in S, i \neq j$. 
		\item $\forall (i,t)\in S: p_i \in P$: fragments can only be retrieved from storage peers. 
	\end{enumerate}
	
	\label{def:restore_schedule} 
\end{definition}

\begin{definition}
	
The \emph{completion time} $C$ of a schedule $S$ is the last time-slot
in which a transfer is performed, that is: $$C(S)= \max \{t:(i,t)\in
S\}.$$
\end{definition}

In the following, we first consider a full information setting, and
show how to compute an optimal schedule which minimizes completion
time provided that the online behavior of peers is known \emph{a
priori}. Then, we compare optimal scheduling to a randomized policy
that needs no knowledge of future peer uptime; via a numeric analysis,
we show the conditions under which a randomized, uninformed approach
achieves performance comparable to that of an optimal schedule.

\subsection{Full Information Setting}

We cast the problem of finding the optimal schedule for both backup
and restore operations as finding the minimum completion time to
transfer a given number $x$ of fragments. For backup, $x$ will
correspond to the number $n$ of redundant encoded fragments; for
restores, $x$ will be equal to the number $k$ of original
fragments. We show that this problem can be reduced to finding the
maximum number of fragments that can be transferred within a given
time $T$. We then use a max-flow formulation and show that existing
algorithms can solve the original problem in polynomial time.

\begin{definition}
	
	An \emph{optimal schedule} to backup/restore $x$ fragments is one that achieves the minimum completion time to transfer at least $x$ fragments. Let $\mathcal S$ be the set of all valid schedules; the minimum completion time is: 
	\begin{equation}
		O(x)=\min\{C(S):S\in \mathcal S\land|S|\geq x\}. \label{eq:mintime} 
	\end{equation}
\end{definition}

The following proposition shows that the optimal completion time can be obtained by computing the maximum number of fragments that can be transferred in $T$ time-slots. 

\begin{proposition}
	Let $\mathcal S$ be the set of all valid schedules and $F(t)$ be the function denoting the maximum number of fragments that can be transferred within time-slot $t$, that is: 
	\begin{equation}
		F(t)=\max\{|S| : S \in \mathcal S \land C(S) \leq t\}. \label{eq:maxfrag} 
	\end{equation}
	The optimal completion time is: $$O(x)=\min\{t : F(t) \geq x\}.$$ 
\end{proposition}
\begin{proof}
	Let $t_1=O(x)$ and $t_2=\min\{ t: F(t) \geq x\}$. 
	\begin{itemize}
		\item $t_1 \geq t_2$. By Eq. \ref{eq:mintime}, an $S_1 \in \mathcal S$ exists such that $C(S_1)=t_1$ and $|S_1| \geq x$, implying that $F(t_1) \geq x$. Therefore, $t_1 \geq \min\{t: F(t) \geq x \}=t_2$.
		
		\item $t_1 \leq t_2$. By Eq. \ref{eq:maxfrag}, an $S_2$ exists such that $C(S_2)=t_2$ and $|S_2| \geq x$. This directly implies that $t_1 = O(x) \leq t_2$. 
	\end{itemize}
\end{proof}

We can now iteratively compute $F(t)$ with growing values of $t$; the
above Proposition guarantees that the first value $T$ that satisfies
$F(T) \geq x$ will be the desired result.

We now focus on a single instance of the problem of finding the
maximum number of fragments $F(T)$ that can be transferred within
time-slot $T$, and show that it can be encoded as a max-flow problem
on a flow network built as follows.  First, we create a bipartite
directed graph $G'=(V',E')$ where $V'=\mathcal{T} \cup \mathcal{P}$;
the elements of $\mathcal{T}=\left\{t_i : i \in 1\ldots T\right\}$
represent time-slots, the elements of $\mathcal{P}=\left\{p_i : i \in
1 \ldots I\right\}$ represent remote peers (only storage nodes for
restores). An edge connects a time-slot to a peer if that peer is
online during that particular time-slot: $E'
= \left\{\left(t_i,p_j\right): t_i \in \mathcal T \land
p_j \in \mathcal P \land a_{i,j}=1\right\}$. Source $s$ and sink $t$
nodes complete the bipartite graph $G'$ and create a flow network
$G=(V,E)$. The source is connected to all the time-slots during which
the data owner $p_0$ is online; all peers are connected to the sink.

The capacities on the edges are defined as follows: edges from the source to time-slots have capacity $u_0$ or $d_0$ (respectively, for backup and restore operations); edges between time-slots and peers have capacity $d_i$ or $u_i$ (respectively, for backup and restore operations); finally, edges between peers and the sink have capacity $m$. Note that in this work we assume individual fragments to be uploaded to distinct peers, hence $m=1$. To simplify presentation, we assume integer capacities $u_k=d_k=1\, \forall k \in [0,I]$.

Fig.~\ref{fig:maxflow} illustrates an example of the whole procedure
described above, for the case of a backup
operation. Fig.~\ref{fig:example_time} shows the online behavior for
time-slots $t_1, \ldots, t_8$ of the data owner and the remote peers
($p_1, p_2, p_3$) that can be selected as remote locations to backup
data. The optimal schedule problem amounts to deciding which remote
peer should be awarded a time-slot to transfer backup fragments, so
that the operation can be completed within the shortest time. This
problem is encoded in the graph of
Fig.~\ref{fig:example_maxflow}. Time-slots and remote peers are
represented by the nodes of the inner bipartite graph. An edge of
capacity 1 connects a time-slot to the set of online peers in that
time-slot, as derived from Fig.~\ref{fig:example_time}. The source
node has an edge of capacity $u_0=1$ to every time-slot in which the
data owner is online (in the figure, $t_4,t_5$ are shaded to remind
$p_0$ is offline): this guarantees that only 1 fragment per time-slot
can be transferred. The sink node has an incident edge with capacity
$m=1$ from every remote peer.

\begin{figure}
\centering
	
				\subfigure[Online behavior]{ \label{fig:example_time} 
				\includegraphics[scale=0.4]{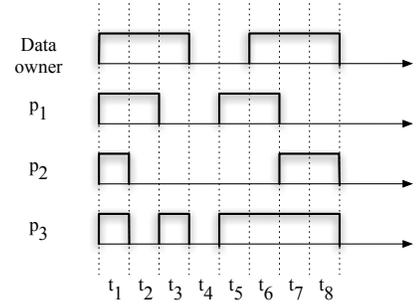}} 
				\subfigure[Equivalent flow network]{ \label{fig:example_maxflow} 
				\includegraphics[scale=0.4]{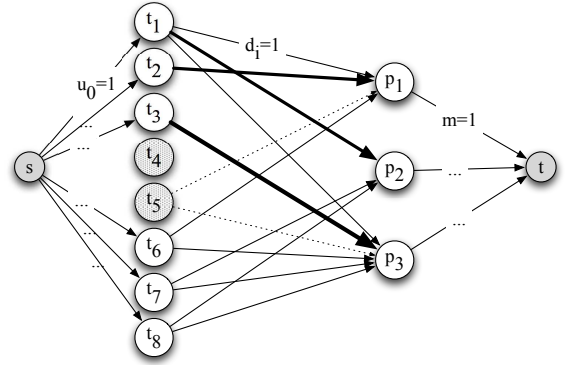}} 
	
	\caption{An example of a backup operation. The original problem of finding an optimal schedule, given the online behavior of peers, is transformed in a max-flow problem on an equivalent graph.} \label{fig:maxflow} 
\end{figure}

Each $s-t$ flow represents a valid schedule. For backup operations,
the schedule is valid because the three equations of
Def. \ref{def:backup_schedule} are verified by construction of the
flow network. Similarly, for restore operations the equations of
Def. \ref{def:restore_schedule} are satisfied by construction, since
only remote peers in the set $P$ are part of the flow network.

In the particular case of the example the smallest value of $t$
ensuring $F(t)\geq 3$ is 3, corresponding to a flow graph that
contains only the $t_1, t_2, t_3$ time-slot nodes. The resulting
optimal scheduling corresponds to the thick edges in
Fig.~\ref{fig:example_maxflow}.

For a flow network with $V$ nodes and $E$ edges, the max-flow can be computed with time complexity $O\left(VE \log\left(\frac{V^2}{E}\right)\right)$ \cite{12144}. In our case, when we have $p$ nodes and an optimal solution of $t$ time-slots, $V$ is $O(p+t)$ and $E$ is $O(pt)$. The complexity of an instance of the algorithm is thus $O\left(pt\left(p \log \frac p t + t \log \frac t p\right)\right)$.

The original problem, i.e., finding an optimal schedule that minimizes
the time to transfer $x$ fragments, can be solved by performing
$O(\log t)$ max-flow computations. In fact, an upper bound for the
optimal completion time can be found in $O(\log t)$ instances of the
max-flow algorithm by doubling at each time the value of $T$, then the
optimal value can be obtained, again in $O(\log t)$ time, by using
binary search. The computational complexity of determining an optimal
schedule in a full information framework is thus $O\left(pt\log
t\left(p \log \frac p t + t \log \frac t p\right)\right)$.

\subsection{Random Scheduling}

In practice, assuming complete knowledge of peers' online behavior is
not realistic. We introduce a \emph{randomized scheduling policy}
which only requires knowing which peers are online at the time of the
scheduling decision. In Sec. \ref{sec:simple_sim}, we compare
optimal and randomized scheduling using real traces.

For backup operations, in each time-slot, fragments are uploaded from
the data owner to no more than $u_0$ remote peers chosen at random
among those that are currently online and that did not receive a
fragment in previous time-slots. This satisfies
Def.~\ref{def:backup_schedule}. For restore operations, in each
time-slot, $d_0$ remote peers in the set $P$ are randomly chosen among
those that are currently online and data is transferred back to the
data owner. This satisfies Def.~\ref{def:restore_schedule}.

We now use Fig.~\ref{fig:maxflow} to illustrate a possible outcome of
the randomized schedule defined here and compare it to the optimal
schedule computed using the max-flow formalization. We focus on the
backup operation of $x=3$ fragments carried out by the data owner
$p_0$. In Fig.~\ref{fig:example_time}, the data owner may randomly
select $p_1$ to be the recipient of the first fragment in time-slot
$t_1$. Since we assume $m=1$ fragment can be stored on a distinct
peer, this choice implies that time-slot $t_2$ is ``wasted''. In
time-slot $t_3$ the data owner has no choice but to store data on peer
$p_3$. Only in time-slot $t_7$ the backup process is complete, when
the last fragment is uploaded to peer $p_2$. Hence, this randomized
schedule writes as $(p_1,t_1);(p_3,t_3);(p_2,t_7)$.

The optimal schedule is obtained by computing the max-flow on the flow network in Fig.~\ref{fig:example_maxflow} (thick edges in the figure), and writes as $(p_2,t_1);(p_1,t_2);(p_3,t_3)$. The backup operation only requires 3 time-slots to complete.

\subsection{Numerical Analysis} \label{sec:simple_sim}

Here, we take a numerical perspective and compare optimal
and randomized scheduling in terms of TTB and TTR. We focus on a
single data owner $p_0$ involved in a backup operation.  The input to
the scheduling problem is the availability trace described in
Sec. \ref{sec:availability_trace}, starting the backup at a random
moment; we set the duration of a time-slot to one hour. Let $u_0=1$
fragment per time-slot be the upload rate of $p_0$. We report results
for $x\in\{40, 60, 80\}$ backup fragments, and vary the number of
randomly chosen remote peers so that $I \in
\{1.1x, 1.2x, \ldots, 2x\}$. We obtained each data point by averaging
1,000 runs of the experiment; furthermore, for each of those runs, we
averaged the completion times of 1,000 random schedules in the same
settings.

\begin{figure}
	\centering
	\includegraphics[width=.8\columnwidth]{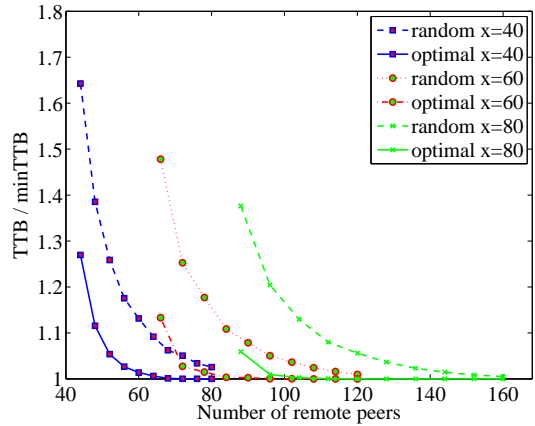} 
	\caption{Numerical analysis: a comparison between optimal and randomized scheduling, using real availability traces.} \label{fig:ineff} 
\end{figure}

Fig.~\ref{fig:ineff} illustrates the ratio between the TTB achieved
respectively by optimal and randomized scheduling, normalized to the
ideal backup time minTTB. We observe that both optimal and randomized
scheduling approach minTTB when the number of remote peers available
to store backup fragments increases: a large system improves
transmission opportunities, and TTB approaches the ideal lower
bound. However, when the number of backup fragments grows, which is a
consequence of higher redundancy rates, randomized scheduling requires
a larger pool of remote machines to approach the performance of the
optimal scheduling. We also note that heterogeneous and correlated
behavior of users in the availability trace results in ``idle''
time-slots in which neither optimal nor randomized scheduling can
transfer data.

This very same evaluation can be used to evaluate a restore operation,
even if the parameters acquire a different meaning. In this case, the
number $x$ of fragments that need to be transferred is the number of
original fragments $k$, and the number of remote peers $I$ will
correspond to the number of encoded fragments $n$. For restores, as
the redundancy rate $\frac n k = \frac I x$ grows, backups will be
more efficient.

We conclude that randomized scheduling is a good choice for a P2P
backup application, provided that:
\begin{itemize}
\item to have efficient backups, the ratio between number of nodes
      in the system and number of fragments to store is not very close
      to one;
\item to have efficient restores, the redundancy rate is not very close to one.
\end{itemize}

As a heuristic threshold, in our analysis we obtain that a value of
$\frac I x=1.5$ is sufficient to complete backup and restore within a
tolerable (around 10\%) deviation from minTTB or minTTR,
respectively. In the following, we will therefore use randomized
scheduling and make sure that such a ratio is reached in order to
ensure that scheduling does not impose a too harsh penalty on TTB and
TTR.

Birk and Kol~\cite{1339108} analyzed random backup scheduling by
modeling peer uptime as a Markovian process. Albeit quantitatively
different due to the absence of diurnal and weekly patterns in their
model, their study reached a conclusion that is analogous to ours: in
backups, the completion time of random scheduling converges to to the
optimal value as the system size grows.

\section{Redundancy Policy}
\label{sec:red}

In the literature, the redundancy rate is generally chosen \emph{a priori} to ensure what we term \emph{prompt data availability}. Given a system with average availability $a$, a target data availability $t$, and assuming the availability of each individual peer as an independent random variable with probability $a$, a \emph{system-wide} redundancy rate is computed as follows. The total number $n$ of redundant fragments required to meet the target $t$, when $k$ original fragments constitute the data to backup is computed as \cite{bhagwan2003availability}:
\begin{equation}
\min \left\{ n \in \mathbb N :
\sum_{i=k}^n {n \choose i} a^i (1 - a)^{n-i}\geq t
\right\}
.
\label{eq:binomial}
\end{equation}
We label this method \emph{fixed-redundancy}, and use it in the following as a baseline approach.

Ensuring prompt data availability is not our goal, since peers only retrieve their data upon (hopefully rare) crash events. Data downloads correspond to restore operations, which require a long time to complete because of the sheer size of backup data.
Hence, we approach the design of our redundancy policy by taking into account the tradeoffs that a backup application has to face. On the one hand, low redundancy improves the aggregate storage capacity of the system, TTB decreases, and maintenance costs drop. On the other hand, two factors discourage from selecting excessively low redundancy rates. First, TTR increases, as less peers will be online to serve fragments during data restores; second, there is a higher risk of data loss.

Our redundancy policy operates as follows. During the backup phase, peers constantly estimate their TTR and the probability of losing data and \emph{adjust} the redundancy rate according to the characteristics of the remote peers that hold their data. In practice, data owners upload encoded fragments until the estimates of TTR and data loss probability are below an arbitrary threshold. When the threshold is crossed, the backup phase terminates.

Note that TTB is generally several times longer than TTR. First, in the restore phase, peers are not likely to disconnect from the Internet. Second, most peers have asymmetric lines with fast downlink and slow uplink; third, backups require uploading redundant data while restores involve downloading an amount of data equivalent to the original backup object. Because of this unbalance, we argue that it is reasonable to use a redundancy scheme that trades longer TTR (which affects only users that suffer a crash) for shorter TTB (which affects all users).

We now delve into the details of how to approximate TTR and data loss probability.

\subsection{Approximating TTR}

Similarly to the optimal scheduling problem, predicting accurately the TTR requires full knowledge of disk failure events and peer availability patterns. We obtain an estimate of the TTR with a heuristic approach; in Sec.~\ref{sec:complex_sim} we show that our approximation is reasonable.

We assume that a data owner $p_0$ remains online during the whole restore process. The TTR can be bounded for two reasons: \emph{i)} the download bandwidth $d_0$ of the data owner is a bottleneck; \emph{ii)} the upload rate of remote peers holding $p_0$'s data is a bottleneck.
Let us focus on the second case: we define the \emph{expected upload rate} of a generic remote peer $p_i$ holding a backup fragment of $p_0$ as the product $a_i u_i$ of the average availability and the upload bandwidth of $p_i$. 
The data owner needs $k$ fragments to recover the backup object: suppose these fragments are served by the $k$ ``fastest'' remote peers. In this case, the ``bottleneck'' upload rate is that of the $k$-th peer $p_j$ with the smaller expected upload rate.
If we consider $l$ parallel downloads and a backup object of size $o$, a peer computes an estimate of the TTR as

\begin{equation}
eTTR = \max \left(\frac o {d_0}, \frac o {l a_j u_j} \right).
\label{eq:eTTR}
\end{equation}

\subsection{Approximating the Data Loss Probability}
\label{sec:dataloss}

Upon a crash, a peer with $n$ fragments placed on remote peers can lose its data if more than $n - k$ of them crash as well before data is completely restored. Considering a delay $w$ that can pass between the crash event and the beginning of the restore phase, we compute the data loss probability within a total delay of $t = w + eTTR$.

We consider disk crashes to be memoryless events, with constant probability for any peer and at any time. Disk lifetimes are thus exponentially distributed stochastic variables with a parametric average $\overline t$: a peer crashes by time $t$ with probability $1 - e^{-t/\overline t}$. The probability of data loss is then
\begin{equation}
\sum_{i=n-k+1}^{n}{n \choose
  i}\left(1-e^{-t/\overline{t}}\right)^{i}\left(e^{-t/\overline{t}}\right)^{n-i}.
\label{eq:dataloss}
\end{equation}

\begin{figure}
\centering
\includegraphics[width=.32\textwidth]{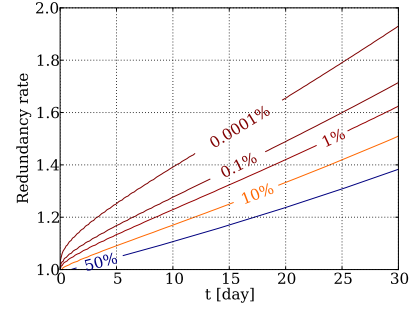}
\caption{Data loss probability.}
\label{fig:dataloss}
\end{figure}

Data loss probability needs to be monitored with great care. In Fig.~\ref{fig:dataloss}, we plot the probability of losing data as a function of the redundancy rate and the delay $t$. Here we set $\overline t=90~\mathrm{days}$ and $k=64$; when the time without maintenance is in the order of magnitude of weeks, even a small decrease in redundancy can increase the probability of data loss by several orders of magnitude.

In summary, our redundancy policy triggers the end of the backup phase, and determines the redundancy rate applied to a backup object. Since we trade longer TTR for shorter TTB, our scheme ensures that data redundancy is enough to make data loss probability small, and keeps TTR under a certain value.
Finally, we remark that our approximation techniques require knowing the uplink capacity and the average availability of remote peers. While a decentralized approach to resource monitoring is an appealing research subject, it is common practice (e.g. Wuala) to rely on a centralized infrastructure to monitor peer resources.

\section{System Simulation}
\label{sec:complex_sim}

We proceed with a trace-driven system simulation, considering all the factors identified in Sec. \ref{sec:scheduling}: churn, correlated uptime, peer bandwidth, congestion, and fragment granularity.

\subsection{Simulation Settings}

Our simulation covers three months, using the availability
traces described in Sec. \ref{sec:availability_trace}, with the
exception that peers remain online during restores. Uplink capacities
of peers are obtained by sampling a real bandwidth distribution
measured at more than 300,000 unique Internet hosts for a 48 hour
period from roughly 3,500 distinct ASes across 160
countries \cite{Piatek07doincentives}. These values have a highly
skewed distribution, with a median of 77 kBps and a mean of
428kBps. To represent typical asymmetric residential Internet lines,
we assign to each peer a downlink speed equal to four times its
uplink.

Our adaptive redundancy policy uses the following parameters: we set the threshold for the estimated TTR to satisfy $eTTR \leq \max\left(1~\mathrm{day}, 2 \cdot minTTR\right)$ and we keep the probability of data loss smaller than $10^{-4}$, when $w=2~\mathrm{weeks}$ is the maximum delay between crash and restore events (see Sec.~\ref{sec:dataloss}).

Each node has 10 GB of data to backup, and dedicates 50 GB of storage space
to the application. The high ratio between these two values
lets us disregard issues due to insufficient storage capacity (which
is considered to be cheap) and focus on the subjects of our
investigation, i.e., scheduling and redundancy. The fragment size $f$
is set to 160 MB, resulting in $k=64$ original fragments per backup
object.

We define peers' lifetimes\footnote{Here we neglect the economics of the application, e.g. promoting user loyalty to the system. Hence, we do not consider unanticipated user departures.} to be exponentially distributed random variables with an expected value of 90 days. After they crash, peers return online immediately and start their restore process; in Sec.~\ref{sec:repair}, we also consider a delay between crash events and restore operations.

As discussed in Sec.~\ref{sec:dataloss}, we compare against a baseline redundancy policy that assigns a fixed redundancy rate. Here we set a target data availability of $t=0.99$, and use the system-wide average availability $a=0.36$ as computed from our availability traces. Hence, we obtain a value $n=228$ and a redundancy rate $n / k = 3.56$. 

Our simulations involve 376 peers. This is sufficient to ensure that the performance of a randomized scheduling is close to optimality (see Sec.~\ref{sec:scheduling}).

For each set of parameters, the simulation results are obtained by
averaging ten simulation runs.

\subsection{Results}

\begin{figure}
\centering

\includegraphics[width=.25\textwidth]{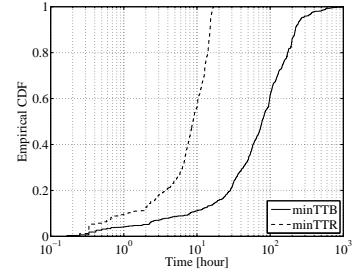}
\caption{MinTTB and minTTR.}
\label{fig:minTTR}
\end{figure}

Fig.~\ref{fig:minTTR} shows the cumulative distribution function of minTTB and minTTR: these baseline values are deeply influenced by the bandwidth distribution we used, and their gap is justified by the asymmetry of the access bandwidth and the assumption that peers stay online during the restore process.

We now verify the accuracy of our approximation of TTR, expressed as the ratio of estimated versus measured TTR. This ratio has a median of $0.92$, with 10th and 90th percentiles of respectively $0.50$ and $2.56$. The values of TTRs vary mostly due to the diurnal and weekly connectivity patterns of users in our traces, but for most cases the eTTR is a sensible rough estimation of TTR.

\begin{figure*}
\centering
	\subfigure[]{
                \includegraphics[width=.3\textwidth]{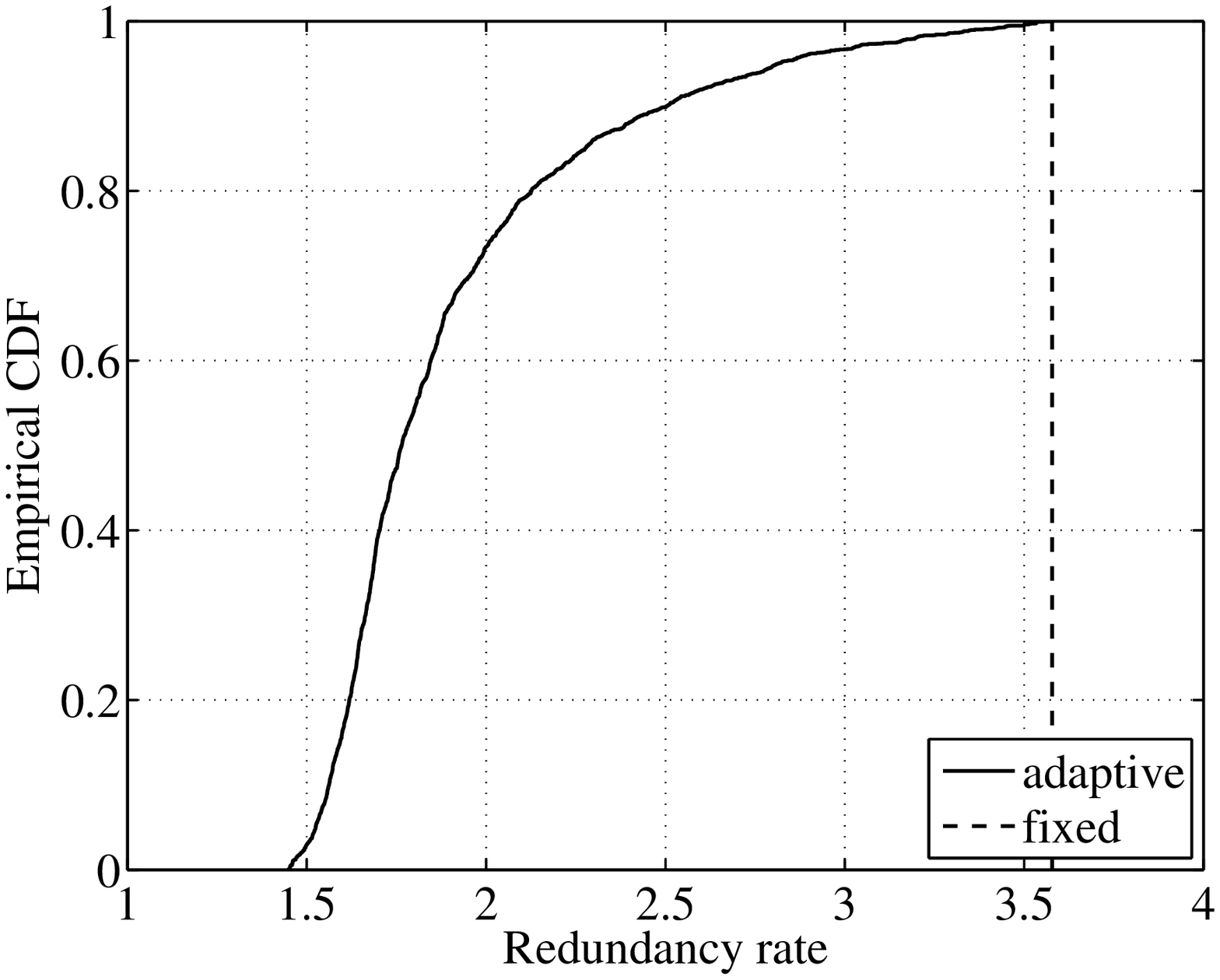}
                \label{fig:redundancy}
        }
	\subfigure[]{
                \includegraphics[width=.3\textwidth]{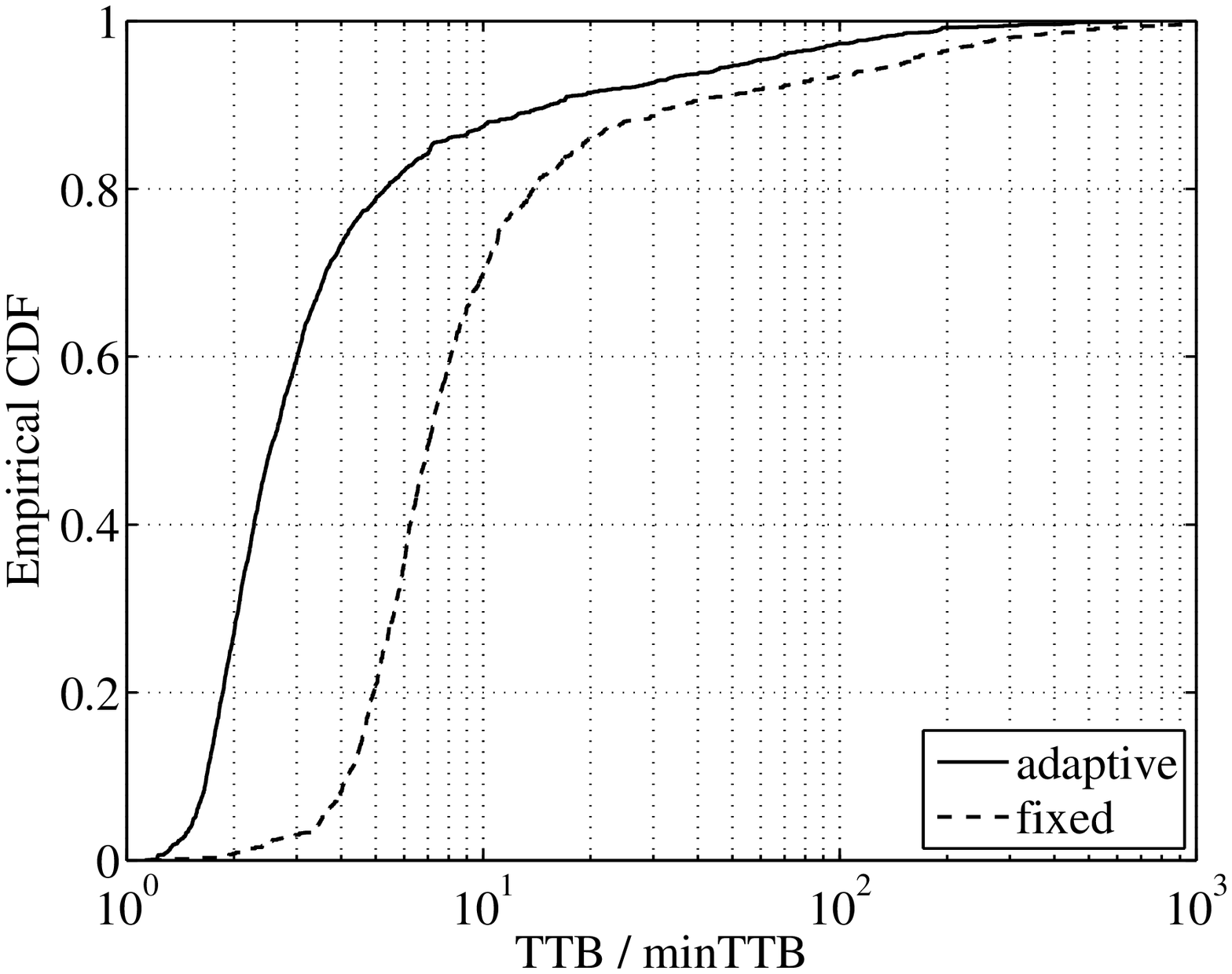}
                \label{fig:ttb}
        }	
	\subfigure[]{
                \includegraphics[width=.3\textwidth]{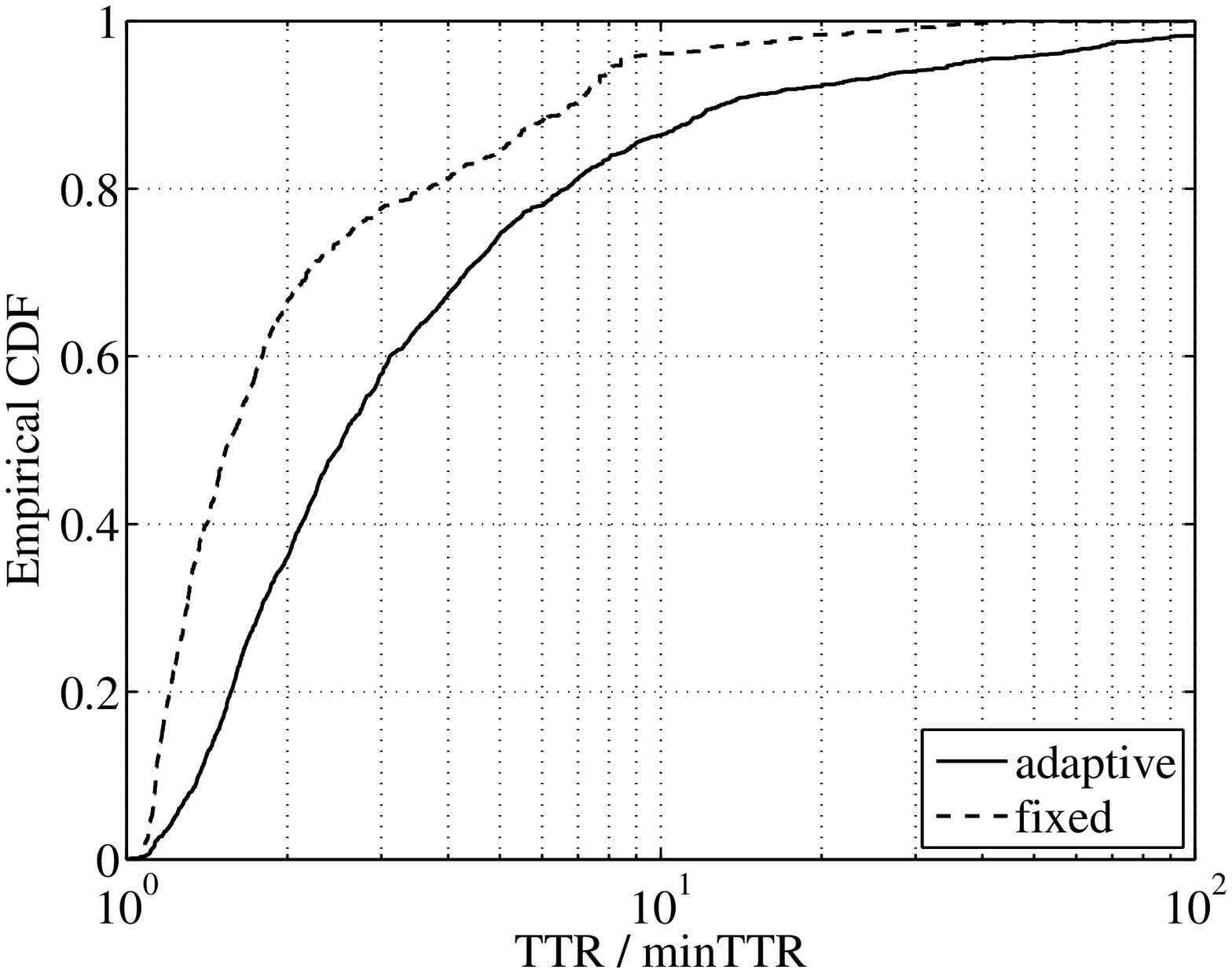}
                \label{fig:ttr}
       }		
\caption{System performance.}
\label{fig:eTTRbased2}
\end{figure*}

The adaptive policy pays off, with an average redundancy rate of
$1.91$ against a flat value of $3.56$ for the baseline approach
(Fig. \ref{fig:redundancy}); the maintenance traffic decreases
accordingly, and the system almost doubles its storage capacity. In
addition, TTB is roughly halved with the adaptive scheme
(Fig. \ref{fig:ttb}); a price for this is paid by crashed peers, which
will have longer TTR (Fig. \ref{fig:ttr}). As we argued in
Sec. \ref{sec:red}, we think this loss is tolerable and well offset by
the benefits of reduced redundancy. We observe
tails where a minority of the nodes have very high $TTB/minTTB$ and
$TTR/minTTR$ ratios. They are nodes with very high bandwidths and
therefore low values of minTTB and minTTR (see Fig.~\ref{fig:minTTR});
their backup and restore speeds will be limited by the bandwidth of
remote nodes which are orders of magnitude smaller.

These results certify that our adaptive scheme beneficially affects performance. However, lower redundancy might result in higher risks of losing data: in the following Section, we analyze this.

\section{Data Loss and Delayed Restore}
\label{sec:repair}

Our simulation settings put the system under exceptional stress: the peer crash rate is two orders of magnitude higher than what is reported for commodity hardware \cite{fast07}. In such an adverse scenario, we study the likelihood and the causes of data loss, and their relation to the redundancy scheme. 

In addition, we discuss the implications of delayed response to crashes, affecting both restore and maintenance operations. We consider the following scenarios:

\begin{itemize}
\item \emph{Immediate response:} Peers start restores as soon as they crash. Moreover, they immediately alert relevant peers to start their maintenance.
\item \emph{Delayed response:} Crashed peers return online after a
      random delay. If this delay exceeds a timeout, peers suffering from fragment loss start their maintenance.
\item \emph{Delayed assisted response:} After the above timeout, a third party intervenes to rescue crashed peers whose data is at risk, by maintaining it.
\end{itemize}

In our simulations, delays are exponentially distributed random variables with an average of one week; the timeout value is one week as well.

For performance reasons, assisted maintenance can be supported by an online storage provider, which is used as a \emph{temporary} buffer. Here we assume a provider with 100\% uptime, unlimited bandwidth and storage space: maintenance is triggered upon expiration of the timeout, conditioned to a data loss probability greater than $10^{-4}$.

\begin{table}
\centering

\caption{Categorization of data loss}
\label{tab:dataloss}
\begin{tabular}{|l|l|r|r|}
\hline
Red. policy & Restores & Unfinished Backups & Unavoidable \\
\hline
\hline

& Immediate & 96\% & 76\% \\

\cline{2-4}

Adaptive & Delayed & 79\% & 65\% \\

\cline{2-4}

& Delayed assisted & 94\% & 78\% \\

\hline

& Immediate & 99\% & 78\% \\

\cline{2-4}

Fixed & Delayed & 92\% & 75\% \\

\cline{2-4}

& Delayed assisted & 94\% & 76\% \\

\hline
\end{tabular}
\end{table}

In our experiments, due to the inflated peer crash rates, between
11.4\% and 14.6\% of crashed peers could not recover their data. In
Table~\ref{tab:dataloss}, we focus on those peers. The majority of
data loss events affected peers that crashed before they completed
their backups, according to the redundancy policy (unfinished backups
column). This can be due to two reasons: the backup process is
inherently time-consuming, due to the availability and bandwidth of
data owners; or the backup system is inefficient.

To differentiate between these two cases, we
consider \emph{unavoidable} data loss events (rightmost column in the
table). If a peer crashes before minTTB, no online backup system could
have saved the data. Data backup takes time: this simple fact alone
accounts for far more than all the limitations of a P2P
approach. Users should worry more about completing their backup
quickly than about the reliability of their peers.

The difference in redundancy between the high rate used by the fixed
baseline and the adaptive approach does not impact significantly
the data loss rate, excepting the case of non-assisted delayed
response. Assisted maintenance is an effective way to counter
this effect.

In Fig.~\ref{fig:outsource} we show the costs of assisted repairs in
terms of data traffic. Given that prices on storage service are highly
asymmetric\footnote{To date (July 2010), inbound traffic to Amazon S3
is free: \url{http://aws.amazon.com/s3/#pricing}.} we only consider
the outbound traffic, from provider to peers.  Data volumes are
expressed as fractions of the total size of backup objects in the
system. There is a striking difference between the adaptive and fixed
redundancy schemes: higher redundancy results in less emergency
situations in which the server has to step in.  The amount of data
stored on the server has a peak load of less than 2.5\% of the total
backup size: the assisted repairs are quick, therefore only a small
fraction of the peers need assistance simultaneously.

\begin{figure}
\centering	
	\includegraphics[width=.25\textwidth]{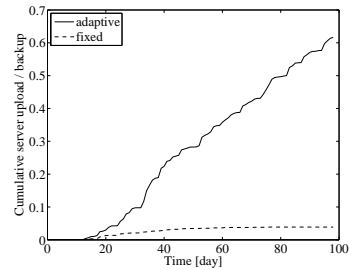}	
\caption{Assisted maintenance.}
\label{fig:outsource}
\end{figure}

\section{Related Work}
\label{sec:related}

Redundancy rates and data repair techniques in P2P backup systems have
been investigated from various angles. Various
works \cite{Chun06efficientreplica, 1774229} determine redundancy as a
function of node failure rate in order to guarantee data durability at
the expense of data availability. Many other approaches
(e.g., \cite{Kubiatowicz00oceanstore, Sameh03acooperative,
Kiran04totalrecall}) adopt formulae similar to
Equation \ref{eq:binomial} to guarantee low latency through prompt
data availability, but require high redundancy rates in typical
settings. Our proposal strives to provide \emph{both} durability and
performance at a low redundancy cost, relaxing prompt data
availability by requiring that data becomes recoverable within a given
time window.

A complete system design requires considering several problems that
were not addressed in this paper; fortunately, many of them have been
tackled in the literature.

When a full system needs to be backed up, \emph{convergent encryption}
\cite{cox02pastiche, 1021936} can be used to ensure that storage space
does not get wasted by saving multiple copies of the same file across
the system.

Data maintenance is cheap in our scenario, where it is performed by a
data owner with a local copy. When maintenance is delegated to nodes
that do not have a local copy of the backup objects, various coding
schemes can be used \cite{dimakis07,duminuco-biersack-08} to
limit the amount of required data transit. For these settings,
cryptographic protocols \cite{oualha2008security,
ateniese2008scalable} have been designed to verify the authenticity of
stored data.

A recurrent problem for P2P applications is
creating incentives to encourage nodes in contributing more
resources. This can be done via reputation
systems \cite{kamvar2003eigentrust} or virtual
currency \cite{vishnumurthy2003karma}. Specifically for storage
systems, an easy and efficient solution is segregating nodes in
sub-networks with roughly homogeneous characteristics such as uptime
and storage space \cite{Pamies-juarez_rewardingstability,
EURECOM+2738}.

Backup objects, whose confidentiality can be ensured by standard
encryption techniques, should encode incremental differences between
archive versions. Recently, various techniques have been proposed to
optimize computational time and size of these
differences \cite{Tangwongsan_efficientsimilarity}.

It may happen that resources offered by peers are just not sufficient
to satisfy all user needs. In this case, a hybrid peer-assisted system
can be developed where data is stored on a centralized data center and
on peers. This can result in scenarios having performances comparable
with centralized systems, at a fraction of the
costs \cite{EURECOM+3140}.

\section{Conclusion}

The P2P paradigm applied to backup applications is a compelling
alternative to centralized online solutions, which become costly
for long-term storage.

In this work, we revisited P2P backup and argued that such an
application is viable. Because the online behavior of users is
unpredictable and, at large scale, crashes and failures are the norm
rather than the exception, we showed that scheduling and redundancy
policies are paramount to achieve short backup and restore times.

We gave a novel formalization of optimal scheduling and showed that,
with full information, a problem that may appear combinatorial in
nature can actually be solved efficiently by reducing it to a maximal
flow problem. Without full information, optimal scheduling is
unfeasible; however, we showed that as the system size grows,
the gap between randomized and optimal scheduling policies diminishes rapidly. 

Furthermore, we studied an adaptive scheme that strives to maintain
data redundancy small, which implies shorter backup times than
a state-of-the-art approach that uses a system-wide, fixed redundancy
rate. This comes at the expense of increased restore times, which we
argued to be a reasonable price to pay, especially in light of our
study on the probability of data loss. In fact, we determined that the
vast majority of data loss episodes are due to incomplete backups. 
Our experiments illustrated that such events are unavoidable, as they
are determined by the limitations of data owners alone: no online
storage system could have avoided such unfortunate events.
We conclude that short backup times are crucial, far more than the
reliability of the P2P system itself. As such, the crux of a P2P
backup application is to design mechanisms that optimize such metric.

Our research agenda includes the design and implementation of a fully
fledged prototype of a P2P backup application. Additionally, we will
extend the parameter space of our study, to include the natural
heterogeneity of user demand in terms of storage requirements. To do
so, we will collect measurements from both existing online storage
systems and from a controlled deployment of our prototype
implementation.

\bibliographystyle{IEEEtran}
\bibliography{infocom}

\end{document}